\documentclass[11pt]{amsart}
\def\isdraft{0}

\usepackage{graphicx,txfonts}
\usepackage{amsaddr} 
\usepackage{mathtools}
\usepackage{amsmath,amssymb,amsfonts,dsfont}
\usepackage{prettyref} 
\usepackage[utf8]{inputenc}
\usepackage{enumitem}
\usepackage[hyphens]{url}
\usepackage{tikz-cd}
\usepackage{tikz}
\usetikzlibrary{positioning}
\usetikzlibrary{arrows}
\usepackage{bussproofs}
\usepackage[mathscr]{euscript}
\usepackage{hyperref} 
\usepackage{amsthm}
\usepackage{theapa}
\usepackage{a4wide}
\usepackage[color=black,textcolor=white\if\isdraft0,disable\fi]{todonotes}
\usepackage{stackengine}
\usepackage{stmaryrd}
\usepackage{wrapfig}
\usepackage{float}
\graphicspath{{fig/}}

\newtheorem{theorem}{Theorem}
\newtheorem{corollary}[theorem]{Corollary}

\newtheorem{lemma}[theorem]{Lemma}

\theoremstyle{definition} 
\newtheorem{definition}[theorem]{Definition}

\newtheorem{remark}[theorem]{Remark}
\newtheorem{example}[theorem]{Example}

\newtheorem{problem}[theorem]{Problem}

\newtheorem{pseudocode}[theorem]{Pseudocode}

\newrefformat{cha}{Chapter \ref{#1}}
\newrefformat{sec}{Section \ref{#1}}
\newrefformat{tab}{Table \ref{#1}}
\newrefformat{fig}{Figure \ref{#1}}
\newrefformat{equ}{(\ref{#1})}
\newrefformat{app}{Appendix \ref{#1}}
\newrefformat{thm}{Theorem \ref{#1}}
\newrefformat{cor}{Corollary \ref{#1}}
\newrefformat{prop}{Proposition \ref{#1}}
\newrefformat{lem}{Lemma \ref{#1}}
\newrefformat{fact}{Fact \ref{#1}}
\newrefformat{obs}{Observation \ref{#1}}
\newrefformat{note}{Note \ref{#1}}
\newrefformat{idea}{Idea \ref{#1}}
\newrefformat{trivia}{Trivia \ref{#1}}
\newrefformat{def}{Definition \ref{#1}}
\newrefformat{not}{Notation \ref{#1}}
\newrefformat{con}{Convention \ref{#1}}
\newrefformat{rem}{Remark \ref{#1}}
\newrefformat{exa}{Example \ref{#1}}
\newrefformat{problem}{Problem \ref{#1}}
\newrefformat{claim}{Claim \ref{#1}}
\newrefformat{conjecture}{Conjecture \ref{#1}}
\newrefformat{exe}{Exercise \ref{#1}}
\newrefformat{alg}{Algorithm \ref{#1}}
\newrefformat{err}{Error \ref{#1}}
\newrefformat{que}{Question \ref{#1}}
\newrefformat{ite}{Item \ref{#1}}
\newrefformat{Q}{Q. \ref{#1}}
\newrefformat{warning}{Warning \ref{#1}}
\newrefformat{pseudocode}{Pseudocode \ref{#1}}

\newcommand{\righttherefore}{:\joinrel\cdot\,}

\title{Logic program proportions}
\author{
	Christian Anti\'c
}
\address{
	christian.antic@icloud.com\\
	Vienna, Austria
}

\begin{document}
\maketitle
\begin{abstract} 
	The purpose of this paper is to present a fresh idea on how symbolic learning might be realized via analogical reasoning. For this, we introduce directed analogical proportions between logic programs of the form ``$P$ transforms into $Q$ as $R$ transforms into $S$'' as a mechanism for deriving similar programs by analogy-making. The idea is to instantiate a fragment of a recently introduced abstract algebraic framework of analogical proportions in the domain of logic programming. Technically, we define proportions in terms of modularity where we derive abstract forms of concrete programs from a ``known'' source domain which can then be instantiated in an ``unknown'' target domain to obtain analogous programs. To this end, we introduce algebraic operations for syntactic logic program composition and concatenation. Interestingly, our work suggests a close relationship between modularity, generalization, and analogy which we believe should be explored further in the future. In a broader sense, this paper is a further step towards a mathematical theory of logic-based analogical reasoning and learning with potential applications to open AI-problems like commonsense reasoning and computational learning and creativity.
\end{abstract}

\section{Introduction}

This paper is a first step towards an answer to the following question:

\begin{quote} 
	How can a computer ``creatively'' generate interesting logic programs from a collection of given ones by using analogical reasoning?
\end{quote}

For example, given a program for the addition of natural numbers, how can we {\em systematically} generate from it a program for ``adding'' lists using analogy? This question can be stated mathematically in the form of a proportional equation between programs as
\begin{align*} 
	Nat:Plus::List:X
\end{align*} where $Nat$ and $List$ are programs for generating the natural numbers and lists, respectively, $Plus$ is the arithmetical program for the addition of numbers, and $X$ is a placeholder for a concrete program, the solution to the equation. That is, we are asking for a program $X=PlusList$ which operates on lists and is analogous to the program $Plus$. A solution to this equation indeed yields a reasonable program for the ``addition'' of lists, namely the program for appending lists (see Examples \ref{exa:mb{Plus}} and \ref{exa:Nat_Plus_List_PlusList}).

To this end, in this paper we introduce directed analogical proportions between logic programs of the form ``$P$ transforms into $Q$ as $R$ transforms into $S$'' --- in symbols, $P\to Q\righttherefore R\to S$ --- as an instance of \citeS{Antic22} general framework of analogical proportions as a mechanism for constructing similar programs by analogy-making. The purpose of this paper is to present a fresh idea on how logic-based symbolic learning can be realized via analogical reasoning, not to give full-fledged solutions to real-world problems --- it is the {\em first} paper in a promising direction, and hopefully not the last one.

In the literature, computational learning usually means learning from (a massive amount of) data. For example, in ``deep learning'' artificial neural networks (ANNs) extract abstract features from data sets \cite<cf.>{Goodfellow16,LeCun15} and, on the symbolic side, inductive logic programs (ILPs) are provided with positive and negative examples of the target concept to be learned \cite<cf.>{Muggleton91}. Another characteristic feature of current machine learning systems is the focus on {\em goal-oriented problem solving} --- a typical task of ANNs is the categorization of the input data (e.g., finding cats in images) and ILPs try to construct logic programs from given examples which partially encode the problem to be solved (e.g., adding numbers or sorting lists).

The emphasis in this paper is different as we believe that {\em program generation} is equally important to artificial intelligence --- and may even be more important for artificial {\em general} intelligence than problem-solving --- and deserves much more attention. This is Sir Michael Atiyah's (fields medalist) answer to the question of how he selects a problem to study:
\begin{quote} 
	I think that presupposes an answer. I don't think that's the way I work at all. Some people may sit back and say, `I want to solve this problem' and they sit down and say, `How do I solve this problem?' I don't. I just move around in the mathematical waters, thinking about things, being curious, interested, talking to people, stirring up ideas; things emerge and I follow them up. Or I see something which connects up with something else I know about, and I try to put them together and things develop. I have practically never started off with any idea of what I'm going to be doing or where it's going to go. I'm interested in mathematics; I talk, I learn, I discuss and then interesting questions simply emerge. I have never started off with a particular goal, except the goal of understanding mathematics. \cite<cf.>{Gowers00}.
\end{quote} The process Sir Michael Atiyah is describing is the generation of new knowledge by connecting existing knowledge in a novel way without any specific ``goal'' in mind and it is believed by many researchers that analogy-making is the core mechanism for doing so \cite<e.g.>{Hofstadter13}. 

In the framework presented in this paper, ``program generation'' means the construction of novel logic programs in an ``unknown'' target domain via {\em analogical transfer} --- realized by directed logic program proportions via generalization and instantiation --- from a ``known'' source domain. This approach is similar to ILP in that novel programs are derived from experience represented as knowledge bases consisting of ``known'' programs. However, it differs significantly from ILP on {\em how} novel programs are constructed from experience --- while in ILP the construction is goal-oriented and thus guided by partial specifications in the form of given examples, in our setting programs are derived by analogy-making to similar programs (without the need for concrete examples). For instance, we may ask --- by analogy to arithmetic --- what it means to ``multiply'' two arbitrary lists (cf. \prettyref{exa:Nat_Plus_List_PlusList}) or to reverse ``even'' lists (cf. Examples \ref{exa:mb{Even}} and \ref{exa:Nat_Even_Reverse_EvenReverse}). Here, contrary to ILP, we do {\em not} expect a supervisor to provide the system with examples explaining list ``multiplication'' or ``evenness'' of lists, but instead we assume that there are arithmetic programs operating on numbers (i.e. numerals) --- programs defining multiplication and evenness of numbers --- which we can transfer to the list domain.

\begin{example}\label{exa:exa} Imagine two domains, one consisting of numbers (or numerals) and the other made up of lists. We know from basic arithmetic what it means to add two numerals. Now suppose we want to transfer the concept of addition to the list domain. We can then ask --- {\em by analogy} --- the following question: What does it mean to ``add'' two lists? We can transform this question into the following directed analogical equation:
\begin{align}\label{equ:Nat_Plus_List_X_1} 
	Nat\to  Plus\righttherefore List\to X.
\end{align} In our framework, $Nat$, $Plus$, and $List$ will be logic programs, and $X$ will be a program variable standing for a program which is obtained from $List$ as $Plus$ is obtained from $Nat$. That is, solutions to \prettyref{equ:Nat_Plus_List_X_1} will be programs implementing `addition of lists'. The idea is to derive an abstract form ${\bf Plus}(Z)$ as a generalization of the concrete program $Plus$ such that 
\begin{align*} 
	Plus={\bf Plus}(Nat).
\end{align*} That is, we factor $Plus$ into subprograms and generalize every instance of the subprogram $Nat$ in $Plus$ by a program variable $Z$. We can then instantiate the form ${\bf Plus}(Z)$ with $List$ to obtain a plausible solution to \prettyref{equ:Nat_Plus_List_X_1}, that is, a program for `addition of lists':
\begin{align*} 
	Nat\to {\bf Plus}(Nat)\righttherefore List\to {\bf Plus}(List).
\end{align*} It is important to emphasize that in order to be able to decompose the program $Plus$ with respect to the above algebraic operations so that $Nat$ occurs as a factor, we need to introduce novel algebraic operations on logic programs (\prettyref{sec:Algebra_of_logic_programs}). We will return to this example, in a more formal manner, in Examples \ref{exa:mb{Plus}} and \ref{exa:Nat_Plus_List_PlusList}.
\end{example}

In (semi-)automatic programming \cite{Czarnecki00},\footnote{\url{https://en.wikipedia.org/wiki/Automatic_programming}} one usually wants to construct a program given some specification. This is complicated by the fact that writing the complete specification is as complex as writing the logic program itself \cite<cf.>{Kowalski84}. In this paper, we therefore propose a differnt view on automatic programming motivated by Sir Atiyah's approach to mathematical research quoted above: Instead of trying to satisfy a (complete) specification, a teacher iteratively constructs a source program $S$ in a ``known'' domain by constructing a sequence of programs $S_1\to S_2\to\ldots\to S_n\to S$ whose ``limit'' is $S$; the student then tries to ``copycat''\footnote{Homage to \citeA{Hofstadter95a}.} the same proceses in an ``unknown'' target domain by constructing the sequence $P_1\to P_2\to\ldots P_n\to P$ such that
\begin{align*} 
	S_1\to S_2 &\righttherefore P_1\to P_2,\qquad S_2\to S_3 \righttherefore P_2\to P_3,\quad\ldots\quad S_n\to S \righttherefore P_n\to P.
\end{align*} By construction, the program $P$ is then ``analogous'' to the source program $S$.

Interestingly, our work suggests a close relationship between modularity, generalization, and analogy which we believe should be explored further in the future.


In a broader sense, this paper is a first step towards a mathematical theory of logic-based analogical reasoning and learning in knowledge representation and reasoning systems with potential applications to fundamental AI-problems like commonsense reasoning and computational learning and creativity.


\section{Logic programs}

In this section, we recall the syntax and semantics of logic programming by mainly following the lines of \citeA{Apt90}.

\subsection{Syntax}

An ({\em unranked first-order}) {\em language} $L$ consists of a set $Ps_L$ of {\em predicate symbols}, a set $Fs_L$ of {\em function symbols}, a set $Cs_L$ of {\em constant symbols}, and a denumerable set $V=\{z_1,z_2,\ldots\}$ of {\em variables}. Terms and atoms are defined in the usual way. Substitutions and (most general) unifiers of terms and (sets of) atoms are defined as usual.

Let $L$ be a language. A ({\em Horn logic}) {\em program} over $L$ (or {\em $L$-program}) is a set of {\em rules} of the form
\begin{align}\label{equ:r} 
	A_0\leftarrow A_1,\ldots,A_k,\quad k\geq 0,
\end{align} where $A_0,\ldots,A_k$ are atoms over $L$. 
It will be convenient to define, for a rule $r$ of the form \prettyref{equ:r}, $head(r):=\{A_0\}$ and $body(r):=\{A_1,\ldots,A_k\}$, extended to programs by $head(P):=\bigcup_{r\in P}head(r)$ and $body(P):=\bigcup_{r\in P}body(r)$. In this case, the {\em size} of $r$ is $k$ denoted by $sz(r)$. A {\em fact} is a rule with empty body and a {\em proper rule} is a rule which is not a fact. We denote the facts and proper rules in $P$ by $facts(P)$ and $proper(P)$, respectively. 

We define the {\em skeleton} inductively as follows: (i) for an atom $p(\vec t)$, define $sk(p(\vec t)):=p$; (ii) for a rule $r$ of the form \prettyref{equ:r}, define 
\begin{align*} 
	sk(r):=sk(A_0)\leftarrow sk(A_1),\ldots,sk(A_k);
\end{align*} finally, (iii) define the skeleton of a program $P$ rule-wise as $sk(P):=\{sk(r)\mid r\in P\}$ (see \prettyref{equ:sk_Tree}). 

A program $P$ is {\em ground} if it contains no variables and we denote the grounding of $P$ which contains all ground instances of the rules in $P$ by $gnd(P)$. 

We call any bijective substitution a {\em renaming}. The set of all {\em variants} of $P$ is defined by $variants(P):=\bigcup_{\theta\text{ renaming}}P[\theta]$. The {\em main predicate} of a program is given by the name of the program in lower case letters if not specified otherwise. We will sometimes write $P\langle p\rangle$ to make the main predicate $p$ in $P$ explicit and we will occasionally write $P(\vec x)$ to indicate that $P$ contains variables among $\vec x=x_1,\ldots,x_n$, $n\geq 1$. We denote the program constructed from $P\langle p\rangle$ by replacing every occurrence of the predicate symbol $p$ with $q$ by $P[p/q]$.

\begin{example}\label{exa:Nat_} Later, we will be interested in the basic data structures of numerals, lists, and (binary) trees. The programs for generating numerals and lists are given by
\begin{align*} 
	Nat(x):=\left\{
	\begin{array}{l}
		nat(0)\\
		nat(s(x))\leftarrow nat(x).
	\end{array}
	\right\} \quad\text{and}\quad List(u,x):=\left\{
	\begin{array}{l}
		list([\;\;])\\
		list([u\mid x])\leftarrow list(x).
	\end{array}
	\right\}.
\end{align*} 

As is customary in logic programming, $[\;\;]$ and $[u\mid x]$ is syntactic sugar for $nil$ and $cons(u,x)$, respectively. The program for generating (binary) trees is given by
\begin{align*} 
	Tree(u,x,y):=\left\{
	\begin{array}{l}
		tree(void)\\
		tree(t(u,x,y))\leftarrow\\ 
		\qquad tree(x),\\
		\qquad tree(y).
	\end{array}
	\right\}.
\end{align*} For instance, the tree consisting of a root $a$, and two leafs $b$ and $c$ is symbolically represented as $tree(t(a,t(b,void,void),t(c,void,void)))$. The skeleton of $Tree$ is given by
\begin{align}\label{equ:sk_Tree} 
	sk(Tree)= \left\{
	\begin{array}{l}
		tree\\
		tree\leftarrow tree
	\end{array}
	\right\}.
\end{align} 

We will frequently refer to the programs above in the rest of the paper.
\end{example}

\subsection{Semantics}

An {\em interpretation} is any set of ground atoms. We define the {\em entailment relation}, for every interpretation $I$, inductively as follows:
\begin{itemize}
	\item For a ground atom $A$, $I\models A$ if $A\in I$.
	\item For a set of ground atoms $B$, $I\models B$ if $B\subseteq I$.
	\item For a ground rule $r$ of the form \prettyref{equ:r}, $I\models r$ if $I\models body(r)$ implies $I\models head(r)$.
	\item Finally, for a ground program $P$, $I\models P$ if $I\models r$ holds for each rule $r\in P$.
\end{itemize} 

In case $I\models gnd(P)$, we call $I$ a {\em model} of $P$. The set of all models of $P$ has a least element with respect to set inclusion called the {\em least model} of $P$ and denoted by $LM(P)$. 

We call a ground atom $A$ a ({\em logical}) {\em consequence} of $P$ --- in symbols $P\models A$ --- if $A$ is contained in the least model of $P$ and we say that $P$ and $R$ are ({\em logically}) {\em equivalent} if $LM(P)=LM(R)$.

\section{Algebra of logic programs}\label{sec:Algebra_of_logic_programs}

Our framework of analogical proportions between logic programs will be built on top of an algebra of logic programs which allows us to decompose programs into simpler modules\footnote{We use here the term ``module'' as a synonym for ``(sub-)program''.} via algebraic operations on programs. For this, it will be useful to introduce in this section two novel algebraic operations for logic program composition (\prettyref{sec:Composition}) and concatenation (\prettyref{sec:Concatenation}).

In the rest of the paper, $P$ and $R$ denote logic programs over some joint unranked first-order language $L$.

\subsection{Composition}\label{sec:Composition}

The rule-like structure of logic programs induces naturally a compositional structure which allows us to decompose programs rule-wise.

We define the ({\em sequential}) {\em composition} of $P$ and $R$ by\footnote{We write $X\subseteq_k Y$ in case $X$ is a subset of $Y$ consisting of $k$ elements.}
\begin{align*} 
	P\circ R:=\left\{head(r\vartheta)\leftarrow body(S\vartheta) \;\middle|\; 
	\begin{array}{l}
		r\in P\\
		S\subseteq_{sz(r)}variants(R)\\
		head(S\vartheta)=body(r\vartheta)\\
		\vartheta=mgu(body(r),head(S))
	\end{array}
	\right\}.
\end{align*}
\todo[inline]{korrekt? gibt es mgu von mengen?}

Roughly, we obtain the composition of $P$ and $R$ by resolving all body atoms in $P$ with the `matching' rule heads of $R$. This is illustrated in the next example, where we construct the even from the natural numbers via composition.

\begin{example}\label{exa:Even} Reconsider the program $Nat$ of \prettyref{exa:Nat_} generating the natural numbers. By composing the only proper rule in $Nat$ with itself, we obtain
\begin{align*} 
	\{nat(&s(x))\leftarrow nat(x)\}\circ\{nat(s(x))\leftarrow nat(x)\}=\{nat(s(s(x)))\leftarrow nat(x)\}.
\end{align*} Notice that this program, together with the single fact in $Nat$, generates the {\em even} numbers. Let us therefore define the program
\begin{align}\label{equ:Even} 
	Even:=\left(facts(Nat)\cup proper(Nat)^2\right)[nat/even]= \left\{
	\begin{array}{l}
		even(0)\\
		even(s(s(x)))\leftarrow\\ 
		\qquad even(x)
	\end{array}
	\right\},
\end{align} where
\begin{align*} 
	proper(Nat)^2=proper(Nat)\circ proper(Nat).
\end{align*} We will come back to this program in \prettyref{exa:mb{Even}}.
\end{example}

The following example shows that, unfortunately, composition is {\em not} associative.

\begin{example} Consider the rule
\begin{align*} 
    r:=a\leftarrow b,c
\end{align*} and the programs
\begin{align*} 
    P:= \left\{
    \begin{array}{l}
        b\leftarrow b\\
        c\leftarrow b,c
    \end{array}
    \right\} \quad\text{and}\quad R:= \left\{
    \begin{array}{l}
        b\leftarrow d\\
        b\leftarrow e\\
        c\leftarrow f
    \end{array}
    \right\}.
\end{align*} Let us compute $(\{r\}P)R$. We first compute $\{r\}P$ by noting that the rule $r$ has two body atoms and has therefore size 2, which means that there is only a single choice of subprogram $S$ of $P$ with two rules, namely $S=P$; this yields
\begin{align*} 
    \{r\}P=\{head(r)\leftarrow body(P)\}=\{a\leftarrow b,c\}=\{r\}.
\end{align*} Next we compute $(\{r\}P)R=\{r\}R$ by noting that now there are two possible $S_1,S_2\subseteq_2 R$ with $head(S_1)=head(S_2)=body(r)$, given by
\begin{align*} 
    S_1= \left\{
    \begin{array}{l}
        b\leftarrow d\\
        c\leftarrow f             
    \end{array}
    \right\} \quad\text{and}\quad S_2= \left\{
    \begin{array}{l}
        b\leftarrow e\\
        c\leftarrow f             
    \end{array}
    \right\}.
\end{align*} This yields
\begin{align*} 
    \{r\}R= \left\{
    \begin{array}{l}
        head(r)\leftarrow body(S_1)\\             
        head(r)\leftarrow body(S_2)             
    \end{array}
    \right\}= \left\{
    \begin{array}{l}
        a\leftarrow d,f\\             
        a\leftarrow e,f             
    \end{array}
    \right\}.
\end{align*} Let us now compute $\{r\}(PR)$. We first compute $PR$. We have
\begin{align*} 
    PR=\{b\leftarrow b\}R\cup \{c\leftarrow b,c\}R.
\end{align*} We easily obtain
\begin{align*} 
    \{b\leftarrow b\}R= \left\{
    \begin{array}{l}
        b\leftarrow d\\             
        b\leftarrow e             
    \end{array}
    \right\}.
\end{align*} Similar computations as above show
\begin{align*} 
    \{c\leftarrow b,c\}R= \left\{
    \begin{array}{l}
        c\leftarrow d,f\\             
        c\leftarrow d,e             
    \end{array}
    \right\}.
\end{align*} So in total we have
\begin{align*} 
    PR= \left\{
    \begin{array}{l}
        b\leftarrow d\\             
        b\leftarrow e\\
        c\leftarrow d,f\\             
        c\leftarrow d,e             
    \end{array}
    \right\}.
\end{align*} To compute $\{r\}(PR)$, we therefore see that there are four two rule subprograms $S_1,S_2,S_4,S_4\subseteq_2 PR$ with $b$ or $c$ in their heads given by
\begin{align*} 
    S_1= \left\{
    \begin{array}{l}
        b\leftarrow d\\
        c\leftarrow d,f             
    \end{array}
    \right\} \qquad S_2= \left\{
    \begin{array}{l}
        b\leftarrow d\\
        c\leftarrow d,e             
    \end{array}
    \right\} \qquad S_3= \left\{
    \begin{array}{l}
        b\leftarrow e\\
        c\leftarrow d,f             
    \end{array}
    \right\} \qquad S_4= \left\{
    \begin{array}{l}
        b\leftarrow e\\
        c\leftarrow d,f             
    \end{array}
    \right\}.
\end{align*} Hence, we have
\begin{align*} 
    \{r\}PR= \left\{
    \begin{array}{l}
        head(r)\leftarrow body(S_1)\\             
        head(r)\leftarrow body(S_2)\\             
        head(r)\leftarrow body(S_3)\\             
        head(r)\leftarrow body(S_4)
    \end{array}
    \right\}= \left\{
    \begin{array}{l}
        a\leftarrow d,f\\
        a\leftarrow d,e\\
        a\leftarrow d,e,f             
    \end{array}
    \right\}.
\end{align*} We have thus shown
\begin{align*} 
    \{r\}(PR)= \left\{
    \begin{array}{l}
        a\leftarrow d,f\\
        a\leftarrow e,f\\
        a\leftarrow d,e,f\\
    \end{array}
    \right\}\neq \left\{
    \begin{array}{l}
        a\leftarrow d,f\\
        a\leftarrow e,f
    \end{array}
    \right\}=(\{r\}P)R.
\end{align*}
\end{example}

\begin{remark}\label{rem:LM_circ} Notice that the least model of $P\circ R$ is, in general, {\em not} obtained from the least models of its factors $P$ and $R$ in an obvious way. For example, the least models of $P:=\{a\leftarrow b\}$ and $R:=\left\{
	\begin{array}{l}
		b\\
		a\leftarrow b
	\end{array}
	\right\}$ are $\emptyset$ and $\{a,b\}$, respectively, but the least model of $P\circ R=\{a\}$ is $\{a\}$.
\end{remark}

\subsection{Concatenation}\label{sec:Concatenation}

In many cases, a program is the `concatenation' of two or more simpler programs on an atomic level. A typical example is the program
\begin{align}\label{equ:Length} 
	Length:=\left\{
	\begin{array}{l}
		length([\;\;],0)\\
		length([u\mid x],s(y))\leftarrow\\ 
		\qquad length(x,y)
	\end{array}\right\},
\end{align} which is, roughly, the `concatenation' of $List$ in the first and $Nat$ in the second argument modulo renaming of predicate symbols (cf. \prettyref{exa:Nat_}). This motivates the following definition.

We define the {\em concatenation} of $P$ and $R$ inductively as follows:
\begin{enumerate}
\item For atoms $p(\vec s)$ and $p(\vec t)$, we define $$p(\vec s)\cdot p(\vec t):=p(\vec s,\vec t),$$ extended to sets of atoms $B$ and $B'$ by
\begin{align*} 
	B\cdot B':=\left\{A\cdot A'\mid A\in B,A'\in B':sk(A)=sk(A')\right\}.
\end{align*}
\item For rules $r$ and $r'$ with $sk(r)=sk(r')$, we define
\begin{align*} 
	r\cdot r':=(head(r)\cdot head(r'))\leftarrow (body(r)\cdot body(r')).
\end{align*}
\item Finally, we define the concatenation of $P$ and $R$ by
\begin{align*} 
	P\cdot R:=\{r\cdot r'\mid r\in P,r'\in R:sk(r)=sk(r')\}.
\end{align*}
\end{enumerate} We will often write $PR$ instead of $P\cdot R$ in case the operation is understood from the context.

We can now formally deconcatenate the list program from above as
\begin{align*}
	Length&=\left\{
	\begin{array}{l}
		length([\;\;])\\
		length([u\mid x])\leftarrow length(x)
	\end{array}
	\right\}\cdot \left\{
	\begin{array}{l}
		length(0)\\
		length(s(y))\leftarrow length(y)
	\end{array}
	\right\}\\
	&= 
	\left\{
	\begin{array}{l}
		list([\;\;])\\
		list([u\mid x])\leftarrow list(x)
	\end{array}
	\right\}[list/length]\cdot \left\{
	\begin{array}{l}
		nat(0)\\
		nat(s(y))\leftarrow nat(y)
	\end{array}
	\right\}[nat/length],
\end{align*} which is equivalent to
\begin{align}\label{equ:List_cdot_Nat}
	Length=List[list/length]\cdot Nat[nat/length].
\end{align} We will return to deconcatenations of this form in \prettyref{sec:Logic_program_forms} (cf. \prettyref{exa:mb{Plus}}).


\begin{theorem} Concatenation is associative.
\end{theorem}
\begin{proof} We know that the concatenation of words is associative. From this we deduce
\begin{align}\label{equ:AA'A''} A\cdot(A'\cdot A'')=(A\cdot A')\cdot A'',\quad\text{for any atoms $A,A',A''$ with $sk(A)=sk(A')=sk(A'')$.}
\end{align} This implies
\begin{align*} B\cdot(B'\cdot B'')&=\left\{A\cdot(A'\cdot A'') \;\middle|\; A\in B,A'\in B',A''\in B'',sk(A)=sk(A')=sk(A'')\right\}\\
&\stackrel{\prettyref{equ:AA'A''}}=\left\{(A\cdot A')\cdot A'' \;\middle|\; A\in B,A'\in B',A''\in B'',sk(A)=sk(A')=sk(A'')\right\}\\
&=(B\cdot B')\cdot B'',\quad\text{for any sets of atoms $B,B',B''$.}
\end{align*} From this, we deduce the associativity of rule concatenation:
\begin{align*} p\cdot (q\cdot r)&=(head(p)\cdot (head(q)\cdot head(r)))\leftarrow (body(p)\cdot (body(q)\cdot body(r)))\\
&=((head(p)\cdot head(q))\cdot head(r))\leftarrow (body(p)\cdot body(q))\cdot body(r))\\
&=(p\cdot q)\cdot r.
\end{align*} We have
\begin{align*} sk(r)=sk(s) \quad\Leftrightarrow\quad sk(r\cdot s)=sk(r)=sk(s),
\end{align*} which finally implies the associativity of concatenation via
\begin{align*} P\cdot(Q\cdot R)&=\bigcup{\substack{p\in P,q\in Q,r\in R\\sk(p)=sk(q)=sk(r)}}(p\cdot (q\cdot r))=\bigcup{\substack{p\in P,q\in Q,r\in R\\sk(p)=sk(q)=sk(r)}}((p\cdot q)\cdot r)=(P\cdot Q)\cdot R.
\end{align*}
\end{proof}

\begin{remark} The least model of $P\cdot R$ is {\em not} obtained from the least models of $P$ and $R$ in an obvious way. For example, we have
\begin{align*} 
	LM(Nat\cdot Nat)=\{nat(s^n(0),s^n(0))\mid n\geq 0\}
\end{align*} whereas
\begin{align*} 
	LM(Nat)\cdot LM(Nat)=\{nat(s^m(0),s^n(0))\mid m,n\geq 0\}.
\end{align*}
\end{remark}

\begin{definition}\label{def:algebra} The {\em algebra of logic programs} over $L$ consists of all $L$-programs together with all unary and binary  operations on programs introduced above, including composition, concatenation, union, and the least model operator.
\end{definition}

\section{Logic program forms}\label{sec:Logic_program_forms}

Recall from \prettyref{exa:exa} that we wish to derive abstract generalizations of concrete programs, which can then be instantiated to obtain similar programs. We formalize this idea via logic program forms as follows.

In the rest of the paper, we assume that we are given {\em program variables} $X,Y,Z,\ldots$ as placeholders for concrete programs and an algebra of logic programs $\mathfrak P$ over some fixed language $L$ (cf. \prettyref{def:algebra}). 

A ({\em logic program}) {\em form} over $\mathfrak P$ (or {\em $\mathfrak P$-form}\todo{$L$-form?}) is any well-formed expression built up from $L$-programs, program variables, and all algebraic operations on programs from above including substitution. More precisely, $\mathfrak P$-forms are defined by the grammar
\begin{align*} 
	\mathbf F::=P\mid Z\mid\mathbf F\cup\mathbf F\mid\mathbf F\circ\mathbf F\mid\mathbf F\cdot\mathbf F\mid\mathbf F\sigma\mid LM(\mathbf F)\mid head(\mathbf F)\mid body(\mathbf F)\mid facts(\mathbf F)\mid proper(\mathbf F),
\end{align*} where $P\in\mathfrak P$ is an $L$-program, $Z$ is a program variable, and $\sigma$ is a substitution. We will denote forms by boldface letters. 

Forms generalize logic programs and induce transformations on programs in the obvious way by replacing program variables with concrete programs. This means that we can interpret logic program forms as `meta-terms' over the algebra of logic programs with programs as `constants', program variables as variables, and algebraic operations on programs as `function symbols'. This is illustrated in the following examples.

\begin{example}\label{exa:mb{Plus}} The program $Plus$ of \prettyref{exa:exa} for the addition of numerals is given by
\begin{align*} 
	Plus:= \left\{
	\begin{array}{l}
		plus(0,y,y)\\
		plus(s(x),y,s(z))\leftarrow\\
		\qquad plus(x,y,z)
	\end{array}
	\right\}.
\end{align*} 

Recall from \prettyref{exa:exa} that we wish to derive a form ${\bf Plus}$ from $Plus$ which abstractly represents addition. Notice that $Plus$ is, essentially, the concatenation of the program $Nat$ (\prettyref{exa:Nat_}) in the first and last argument together with a middle part. Formally, we have
\begin{align*} 
	Plus&= 
		\left\{
			\begin{array}{l}
				plus(0)\\
				plus(s(x))\leftarrow\\
				\qquad plus(x)		
			\end{array}
		\right\}\cdot \left\{
			\begin{array}{l}
				plus(y,y)\\
			  	plus(y)\leftarrow\\ 
			  	\qquad plus(y)
			\end{array}
		\right\}\cdot \left\{
			\begin{array}{l}
				plus()\\
				plus(s(z))\leftarrow\\
				\qquad plus(z)		
			\end{array}
		\right\}\\
		&=\left\{
			\begin{array}{l}
				nat(0)\\
				nat(s(x))\leftarrow\\
				\qquad nat(x)		
			\end{array}
		\right\}[nat/plus]\cdot \left\{
			\begin{array}{l}
				plus(y,y)\\
			  	plus(y)\leftarrow\\ 
			  	\qquad plus(y)		
			\end{array}
		\right\}\cdot \left\{
			\begin{array}{l}
				plus()\\
				nat(s(z))\leftarrow\\
				\qquad nat(z)		
			\end{array}
		\right\}[nat/plus]\\
		&=Nat(x)[nat/plus]\cdot\left\{
			\begin{array}{l}
				plus(y,y)\\
			  	plus(y)\leftarrow\\ 
			  	\qquad plus(y)
			\end{array}
		\right\}\cdot\left\{
			\begin{array}{l}
			  	plus()\\
			  	proper(Nat(z))[nat/plus]
			\end{array}
		\right\}.
\end{align*} 

We therefore define the form ${\bf Plus}(Z\langle q\rangle(\vec x))$, where $Z$ is a program variable, $q$ stands for the main predicate symbol in $Z$, and $\vec x$ is a sequence of variables, by
\begin{align*} 
	{\bf Plus}(Z\langle q\rangle(\vec x)):=
		Z[q/plus]\cdot
		\left\{
			\begin{array}{l}
			  plus(y,y)\\
			  plus(y)\leftarrow\\ 
			  \qquad plus(y)
			\end{array}
		\right\}\cdot\left\{
			\begin{array}{l}
			  plus()\\
			  proper(Z)[q/plus,\vec x/\vec z]
			\end{array}
		\right\}.
\end{align*} Here $\vec z$ is a sequence of fresh variables distinct from the variables in $\vec x$. We can think of ${\bf Plus}$ as a generalization of $Plus$ where we have abstracted from the concrete data type $Nat$. In fact, $Plus$ is an instance of ${\bf Plus}$:
\begin{align*} Plus={\bf Plus}(Nat(x)).
\end{align*} 

Similarly, instantiating the form ${\bf Plus}$ with the program $List(u,x)$ for constructing the data type of lists (cf. \prettyref{exa:Nat_}) yields the program\footnote{Here we have instantiated the sequence of variables $\vec x$ with $\vec z=(u,z)$.}
\begin{align*}
	&plus([\;\,],y,y)\\
	&plus([u\mid x],y,[u\mid z])\leftarrow\\
	&\qquad plus(x,y,z),
\end{align*} which is the program for appending lists. For example, we have
\begin{align*} {\bf Plus}(List)\models plus([a,b],[c,d],[a,b,c,d]).
\end{align*}

As a further example, we want to define the `addition' of (binary) trees by instantiating the form ${\bf Plus}$ with $Tree$. Note that we now have multiple choices: since $Tree(u,x,y)$ contains {\em two} variables $x$ and $y$ occurring in the second rule's body and head, we have two possibilities: we can either choose $x=y$ or $x\neq y$. Let us first consider the program 
\begin{align}\label{equ:PlusTree} {\bf Plus}(Tree(u,x,x))= \left\{
\begin{array}{l}
	plus(void,y,y)\\
	plus(t(u,x,x),y,t(u,z,z))\leftarrow\\
	\qquad plus(x,y,z).
\end{array}
\right\}.
\end{align} This program `appends' the tree in the second argument to each leaf of the {\em symmetric} tree in the first argument. Notice that all of the above programs are syntactically almost identical, e.g., we can transform $Plus$ into ${\bf Plus}(List)$ via a simple rewriting of terms. The next program shows that we can derive programs from ${\bf Plus}$ which syntactically differ more substantially from the above programs. Concretely, the program 
\begin{align*} {\bf Plus}(Tree(u,x_1,x_2))= \left\{
\begin{array}{l}
	plus(void,y,y)\\
	plus(t(u,x_1,x_2),y,t(u,z_1,z_2))\leftarrow\\
	\qquad plus(x_1,y,z_1),\\
	\qquad plus(x_2,y,z_1),\\
	\qquad plus(x_1,y,z_2),\\
	\qquad plus(x_2,y,z_2).
\end{array}
\right\}
\end{align*} is logically equivalent to program \prettyref{equ:PlusTree}. However, in some situations this more complicated representation is beneficial. For example, we can now remove the second and third body atom to obtain the more compact program
\begin{align*} 
&plus(void,y,y)\\
&plus(t(u,x_1,x_2),y,t(u,z_1,z_2))\leftarrow\\
&\qquad plus(x_1,y,z_1),\\
&\qquad plus(x_2,y,z_2).
\end{align*} This program, in analogy to program \prettyref{equ:PlusTree}, `appends' the tree in the second argument to each leaf of the not necessarily symmetric tree in the first argument and thus generalizes \prettyref{equ:PlusTree}. Generally speaking, solutions to proportional equations my be inexact in nature needing further transformation in order to satisfy additional information and constraints.
\end{example}

\begin{example}\label{exa:mb{Even}} In \prettyref{exa:Even}, we have constructed the program $Even$, representing the even numbers, from $Nat$ by inheriting its fact and by iterating its proper rule once. By replacing $Nat$ in \prettyref{equ:Even} by a program variable $Z$, we arrive at the form
\begin{align}\label{equ:mb{Even}} \mathbf{Even}(Z):=facts(Z)\cup (proper(Z)\circ proper(Z)).
\end{align} We can now instantiate this form with arbitrary programs to transfer the concept of ``evenness'' to other domains. For example, consider the program $Reverse$ for reversing lists given by $$Reverse:=Reverse_0\cup {\bf Plus}(List(u,x)),$$ where
\begin{align*} Reverse_0:= \left\{
\begin{array}{l}
	reverse([\;\,],[\;\,])\\
	reverse([u\mid x],y)\leftarrow\\
	\qquad reverse(x,z),\\
	\qquad plus(z,[u],y)
\end{array}
\right\}.
\end{align*} By instantiating the form $\mathbf{Even}$ with $Reverse$, we obtain the program
\begin{align*} \mathbf{Even}(Reverse)= \left\{
\begin{array}{l}
	reverse([\;\,],[\;\,])\\
	reverse([u_1,u_2\mid x],[u_3\mid y])\leftarrow\\
	\qquad reverse(x,z),\\
	\qquad plus(z,[u_2],[u_3\mid w]),\\
	\qquad plus(w,[u_1],y),\\
	plus([\;\,],y,y)\\
	plus([u_1,u_2\mid x],y,[u_1,u_2\mid z])\leftarrow\\
	\qquad plus(x,y,z).
\end{array}
\right\}.
\end{align*} One can verify that this program reverses lists of {\em even} length. Similarly, if $Sort$ is a program for sorting lists, then ${\bf Even}(Sort)$ is a program for sorting ``even'' lists and so on.
\end{example}

\begin{example}\label{exa:mb{Member}} The program for checking list membership is given by
\begin{align*} Member:= \left\{
\begin{array}{l}
	member(u,[u\mid x])\\
	member(u,[v\mid x])\leftarrow\\
	\qquad member(u,x)
\end{array}
\right\}.
\end{align*} 

Notice the syntactic similarity between the program $List$ of \prettyref{exa:Nat_} and the second arguments in $Member$ --- in fact, we can deconcatenate $Member$ as follows:
\begin{align*} Member= \left\{
\begin{array}{l}
	member(u)\\
	member(u)\leftarrow\\
	\qquad member(u)
\end{array}
\right\}\cdot \left\{
\begin{array}{l}
	member([u\mid x])\\
	member([v\mid x])\leftarrow\\
	\qquad member(x)
\end{array}
\right\}.
\end{align*} The second factor can be expressed in terms of $List$ via
\begin{align*} \{member([u\mid x])\}=\left(proper(List(u,x))\circ body(proper(List(u,x)))\right)[list/member]
\end{align*} and
\begin{align*} \{member([v\mid x])\leftarrow member(x)\}=proper(List(v,x))[list/member].
\end{align*} This yields the form ${\bf Member}(Z(u,\vec x)\langle q\rangle)$, where $\vec x$ is a (possibly empty) sequence of variables, given by\todo{$\vec x$ kommt nicht vor}
\begin{align*} \left\{
\begin{array}{l}
	member(u)\\
	member(u)\leftarrow\\
	\qquad member(u)
\end{array}
\right\}\cdot \left\{
\begin{array}{l}
	proper(Z(u,x))\circ body(proper(Z(u,x)))\\
	proper(Z(u,x))[u/v]
\end{array}
\right\}[q/member].
\end{align*} 

We can now ask --- by analogy --- what ``membership'' means in the numerical domain. For this, we compute
\begin{align*} {\bf Member}(Nat(u))= \left\{
\begin{array}{l}
	member(u,s(u))\\
	member(u,s(v))\leftarrow\\
	\qquad member(u,v)
\end{array}
\right\}.
\end{align*} One can easily check that this program computes the ``less than'' relation between numerals.
\end{example}

\section{Logic program proportions}

This is the main section of the paper. Recall from \prettyref{exa:exa} that we want to formalize analogical reasoning and learning in the logic programming setting via directed analogical proportions between programs. For this, we instantiate here a fragment of \citeS{Antic22} abstract algebraic framework of analogical proportions within the algebra of logic programs from above using logic program forms.

Let us first recall \citeS{Antic22} framework, where we restrict ourselves to the directed fragment. In the rest of the paper, we may assume some ``known'' source domain $\mathfrak P$ and some ``unknown'' target domain $\mathfrak R$, both algebras of logic programs over some languages $L_\mathfrak P$ and $L_\mathfrak R$, respectively. We may think of the source domain $\mathfrak P$ as our background knowledge --- a repertoire of programs we are familiar with --- whereas $\mathfrak R$ stands for an unfamiliar domain which we want to explore via analogical transfer from $\mathfrak P$. For this we will consider directed analogical equations of the form `$P$ transforms into $Q$ as $R$ transforms into $X$' --- in symbols, $P\to Q\righttherefore R\to X$ --- where $P$ and $Q$ are programs of $\mathfrak P$, $R$ is a program of $\mathfrak R$, and $X$ is a program variable. The task of learning logic programs by analogy is then to solve such equations and thus to expand our knowledge about the intimate relationships between (seemingly unrelated) programs, that is, solutions to directed analogical equations will be programs of $\mathfrak R$ which are obtained from $R$ in $\mathfrak R$ as $Q$ is obtained from $P$ in $\mathfrak P$ in a mathematically precise way (\prettyref{def:PQRZ}). Specifically, we want to functionally relate programs via rewrite rules as follows. Recall from \prettyref{exa:exa} that transforming $Nat$ into $Plus$ means transforming ${\bf Id}(Nat)$ into ${\bf Plus}(Nat)$, where ${\bf Id}(Z):=Z$ and ${\bf Plus}(Z)$ are forms. We can state this transformation more pictorially as the rewrite rule ${\bf Id}\to{\bf Plus}$. Now transforming the program $List$ `in the same way' means to transform ${\bf Id}(List)$ into ${\bf Plus}(List)$, which again is an instance of ${\bf Id}\to{\bf Plus}$. Let us make this notation official. We will always write ${\bf F}(\vec Z)\to{\bf G}(\vec Z)$ or $\bf F\to G$ instead of $(\bf{F,G})$, for any pair of forms $\mathbf F$ and $\mathbf G$ containing program variables among $\vec Z$ such that every program variable in $\mathbf G$ occurs in $\mathbf F$. We call such expressions {\em justifications}. We denote the set of all justifications with variables among $\vec Z$ by $J(\vec Z)$. We make the convention that $\to$ binds weaker than every other algebraic operation.

The above explanation motivates the following definition. Define the {\em set of justifications} of two programs $P$ and $R$ in $\mathfrak P$ by
\begin{align*} 
	Jus_\mathfrak P(P\to R):=\left\{\mathbf F\to\mathbf G\in J(\vec Z) \;\middle|\; P\to R=\mathbf F(\vec O)\to\mathbf G(\vec O),\text{ for some }\vec O\in\mathfrak P^{|\vec Z|}\right\}.
\end{align*} For instance, $Jus(Nat\to {\bf Plus}(Nat))$ and $Jus(List\to {\bf Plus}(List))$ both contain the justification $Z\to{\bf Plus}(Z)$.

We are now ready to state the main definition of the paper as an instance of (the directed fragment of) \citeS[Definition 5]{Antic22}. 

\begin{definition}\label{def:PQRZ} A {\em directed program equation} in $(\mathfrak{P,R})$ is an expression of the form `$P$ transforms into $Q$ as $R$ transforms into $X$' --- in symbols,
\begin{align}\label{equ:PQRZ} 
	P\to Q\righttherefore R\to X,
\end{align} where $P$ and $Q$ are source programs from $\mathfrak P$, $R$ is a target program from $\mathfrak R$, and $X$ is a program variable. 

Given a target program $S\in\mathfrak R$, define the set of {\em justifications} of $P\to Q\righttherefore R\to S$ in $(\mathfrak{P,R})$ by
\begin{align*} Jus_{(\mathfrak{P,R})}(P\to Q\righttherefore R\to S):=Jus_\mathfrak P(P\to Q)\cap Jus_\mathfrak R(R\to S).
\end{align*} 

We say that $J$ is a {\em trivial set of justifications} in $(\mathfrak{P,R})$ iff every justification in $J$ justifies every directed proportion $P\to Q\righttherefore R\to S$ in $(\mathfrak{P,R})$, that is, iff
\begin{align*} 
    J\subseteq Jus_{(\mathfrak{P,R})}(P\to Q\righttherefore R\to S)\quad\text{for all $P,Q\in\mathfrak P$ and $R,S\in\mathfrak R$.}
\end{align*} In this case, we call every justification in $J$ a {\em trivial justification} in $(\mathfrak{P,R})$. 

Now we call $S$ a {\em solution} to \prettyref{equ:PQRZ} in $(\mathfrak{P,R})$ iff either $Jus_\mathfrak P(P\to Q)\cup Jus_\mathfrak R(R\to S)$ consists only of trivial justifications, in which case there is neither a non-trivial transformation of $P$ into $Q$ in $\mathfrak P$ nor of $R$ into $S$ in $\mathfrak R$; or $Jus_{(\mathfrak{P,R})}(P\to Q\righttherefore R\to S)$ is maximal with respect to subset inclusion among the sets $Jus_{(\mathfrak{P,R})}(P\to Q\righttherefore R\to S')$, $S'\in\mathfrak R$, containing at least one non-trivial justification, that is, for any program $S'\in \mathfrak R$,
\begin{align*} Jus_{(\mathfrak{P,R})}(P\to Q\righttherefore R\to S)&\subseteq Jus_{(\mathfrak{P,R})}(P\to Q\righttherefore R\to S')
\end{align*} implies
\begin{align*} Jus_{(\mathfrak{P,R})}(P\to Q\righttherefore R\to S')\subseteq Jus_{(\mathfrak{P,R})}(P\to Q\righttherefore R\to S).
\end{align*} In this case, we say that $P,Q,R,S$ are in {\em directed logic program proportion} in $(\mathfrak{P,R})$ written as
\begin{align*} (\mathfrak{P,R})\models P\to Q\righttherefore R\to S.
\end{align*} 

We denote the set of all solutions to \prettyref{equ:PQRZ} in $(\mathfrak{P,R})$ by $Sol_{(\mathfrak{P,R})}(P\to Q\righttherefore R\to X)$.
\end{definition}

Roughly, a program $S$ in the target domain is a solution to a directed program equation of the form $P\to Q\righttherefore R\to X$ iff there is no other target program $S'$ whose transformation from $R$ is more similar to the transformation of $P$ into $Q$ in the source domain expressed in terms of maximal sets of algebraic justifications.

We will always write $\mathfrak P$ instead of $(\mathfrak{P,P})$. In what follows, we will usually omit trivial justifications from notation. So, for example, we will write $Jus_{(\mathfrak{P,R})}(P\to Q\righttherefore R\to S)=\emptyset$ instead of $Jus_{(\mathfrak{P,R})}(P\to Q\righttherefore R\to S)=\{\text{trivial justifications}\}$ in case $P\to Q\righttherefore R\to S$ has only trivial justifications in $(\mathfrak{P,R})$, et cetera. The empty set is always a trivial set of justifications. Every justification is meant to be non-trivial unless stated otherwise. 

The forms
\begin{align*} 
	tr_1(X,Y):=(X\cap Y)\cup (X-Y) \quad\text{and}\quad tr_2(X,Y):=(X\cap Y)\cup (Y-X)
\end{align*} justify {\em any} proportion $P\to Q\righttherefore R\to S$, which shows that $tr_1\to tr_2$ is a trivial justification. This example shows that trivial justifications may contain useful information about the underlying structures --- in this case, it encodes the trivial observation that any two programs $P$ and $Q$ are symmetrically related via $P=(P\cap Q)\cup (P-Q)$ and $Q=(P\cap Q)\cup (Q-P)$.

We call a form $\mathbf F(\vec Z)$ a {\em $\mathfrak P$-generalization} of a program $P$ in $\mathfrak P$ iff $P=\mathbf F(\vec O)$, for some $\vec O\in\mathfrak P^{|\vec Z|}$, and we denote the set of all $\mathfrak P$-generalizations of $P$ in $\mathfrak P$ by $Gen_\mathfrak P(P)$. Moreover, we define for any programs $P\in\mathfrak P$ and $R\in\mathfrak R$:
\begin{align*} 
    Gen_{(\mathfrak{P,R})}(P,R):=Gen_\mathfrak P(P)\cap Gen_\mathfrak R(R).
\end{align*}



\begin{example} Consider the directed equation of \prettyref{exa:exa} given by
\begin{align}\label{equ:Nat_Plus_List_Z} Nat\to Plus\righttherefore List\to X.
\end{align} This equation asks for a list program $S$ which is obtained from $List$ as the program $Plus$ on numerals is obtained from $Nat$. In \prettyref{exa:Nat_Plus_List_PlusList}, we will see that the program for concatenating lists is a solution to \prettyref{equ:Nat_Plus_List_Z}.
\end{example}

\begin{example} Consider the directed equation given by
\begin{align}\label{equ:Nat_Even_Reverse_Z} Nat\to Even\righttherefore Reverse\to X,
\end{align} where $Reverse$ is the program for reversing lists of \prettyref{exa:mb{Even}}. In \prettyref{exa:Nat_Even_Reverse_EvenReverse}, we will see that the program for reversing lists of even length is a solution to \prettyref{equ:Nat_Even_Reverse_Z}.
\end{example}

To guide the AI-practitioner, we shall now rewrite the above framework in a more algorithmic style \cite<cf.>[Pseudocode 17]{Antic22}.

\begin{pseudocode}\label{pseudocode:Sol} Computing the solution set $\mathcal S$ to a directed logic program equation $P\to Q \righttherefore R\to X$ consists of the following steps:
\begin{enumerate}
\item Compute $\mathcal S_0:=Sol_{(\mathfrak{P,R})}(P\to Q\righttherefore R\to X)$:
    \begin{enumerate}
    \item For each $S\in\mathfrak R$, if $Jus_\mathfrak P(P,Q)\cup Jus_\mathfrak R(R,S)$ consists only of trivial justifications, then add $S$ to $\mathcal S_0$.

    \item For each form $\mathbf F(\vec Z)\in Gen_{(\mathfrak{P,R})}(P,R)$ and all witnesses $\vec O_1\in\mathfrak P^{|\vec Z|},\vec O_2\in\mathfrak R^{|\vec Z|}$ satisfying
    \begin{align*}
        P=\mathbf F(\vec O_1) \quad\text{and}\quad R=F(\vec O_2),
    \end{align*} and for each form $\mathbf G(\vec Z)\in Gen_\mathfrak P(Q)$ containing only variables occurring in $\mathbf F(\vec Z)$ and satisfying
    \begin{align*} 
        Q=\mathbf G(\vec O_1),
    \end{align*} add $\mathbf{F\to G}$ to $Jus_{(\mathfrak{P,R})}(P\to Q\righttherefore R\to\mathbf G(\vec O_2))$.

    \item Identify those non-empty sets $Jus_{(\mathfrak{P,R})}(P\to Q\righttherefore R\to S)$ which are subset maximal with respect to $S$ and add those $S$'s to $\mathcal S_0$.
    \end{enumerate}
\item For each $S\in\mathcal S_0$, check the following relations with the above procedure:
    \begin{enumerate}
    \item $R\in Sol_{(\mathfrak{P,R})}(Q\to P\righttherefore S\to X)$?
    \item $Q\in Sol_{(\mathfrak{B,A})}(R\to S\righttherefore P\to X)$?
    \item $P\in Sol_{(\mathfrak{B,A})}(S\to R\righttherefore Q\to X)$?
    \end{enumerate} Add those $S\in\mathcal S_0$ to $\mathcal S$ which pass all three tests. The set $\mathcal S$ now contains all solutions to $P:Q::R:X$ in $(\mathfrak{P,R})$.
\end{enumerate}
\end{pseudocode}

\section{Properties of logic program proportions}\label{sec:Properties_}

We summarize here \citeS{Antic22} most important properties of analogical equations and proportions interpreted in the logic programming setting from above.

\subsection{Characteristic justifications}

Computing all justifications of an analogical proportion is complicated in general, which fortunately can be omitted in many cases.

We call a set $J$ of justifications a {\em characteristic set of justifications} \cite[Definition 20]{Antic22} of $P\to Q\righttherefore R\to S$ in $(\mathfrak{P,R})$ iff $J$ is a sufficient set of justifications of $P\to Q\righttherefore R\to S$ in $(\mathfrak{P,R})$, that is, iff
\begin{enumerate}
	\item $J\subseteq Jus_{(\mathfrak{P,R})}(P\to Q\righttherefore R\to S)$, and
	\item $J\subseteq Jus_{(\mathfrak{P,R})}(P\to Q\righttherefore R\to S')$ implies $S'=S$, for each $S'\in\mathfrak P$.
\end{enumerate} In case $J=\{\mathbf F\to\mathbf G\}$ is a singleton, we call $\mathbf F\to\mathbf G$ a {\em characteristic justification} of $P\to Q\righttherefore R\to S$ in $(\mathfrak{P,R})$. 



The following lemma is a useful characterization of characteristic justifications in terms of mild injectivity \cite<cf.>[Uniqueness Lemma]{Antic22}.

\begin{lemma}[Uniqueness Lemma]\label{lem:Uniqueness_Lemma} For any justification $\mathbf F(\vec Z)\to\mathbf G(\vec Z)$ of $P\to Q\righttherefore R\to S$ in $(\mathfrak{P,R})$, if there is a unique $\vec O\in\mathfrak R^{|\vec Z|}$ such that $R=\mathbf F(\vec O)$, then $\mathbf F\to\mathbf G$ is a characteristic justification of $P\to Q\righttherefore R\to S$ in $(\mathfrak{P,R})$.
\end{lemma}
\begin{proof} See the proof of \citeS[Uniqueness Lemma]{Antic22}.
\end{proof}

\subsection{Functional proportion theorem}\label{sec:Functional_proportion_theorem}

In the rest of the paper, we will often use the following reasoning pattern which roughly says that {\em functional dependencies} are preserved across (different) domains \cite<cf.>[Functional Proportion Theorem]{Antic22}:

\begin{theorem}[Functional Proportion Theorem]\label{thm:Functional_Proportion_Theorem} For any $(\mathfrak{A\cap B})$-form $\mathbf G(Z)$, 
we have
\begin{align*} 
	(\mathfrak{P,R})\models P\to \mathbf G(P)\righttherefore R\to \mathbf G(R),\quad\text{for all $P\in\mathfrak P$ and $R\in\mathfrak R$}.
\end{align*} In this case, we call $G(R)$ a {\em functional solution} of $P\to Q \righttherefore R\to X$ in $(\mathfrak{P,R})$ {\em characteristically justified by $Z\to\mathbf G(Z)$}.
\end{theorem}
\begin{proof} See the proof of \citeS[Functional Proportion Theorem]{Antic22}.
\end{proof}

Functional solutions are plausible since transforming $P$ into $\mathbf G(P)$ and $R$ into $\mathbf G(R)$ is a direct implementation of ``transforming $P$ and $R$ in the same way'', and it is therefore surprising that functional solutions can be nonetheless `unexpected' and therefore `creative' as will be demonstrated in \prettyref{sec:Examples}.

\begin{remark}\label{rem:Q} An interesting consequence of \prettyref{thm:Functional_Proportion_Theorem} is that in case $Q\in \mathfrak P\cap\mathfrak R$ is a constant program contained in both domains $\mathfrak P$ and $\mathfrak R$, we have
\begin{align}\label{equ:PTRT} 
	(\mathfrak{P,R})\models P\to Q\righttherefore R\to Q,\quad\text{for {\em all} $P\in\mathfrak P$ and $R\in\mathfrak R$},
\end{align} characteristically justified by \prettyref{thm:Functional_Proportion_Theorem} via $Z\to Q$. This can be intuitively interpreted as follows: every program in $\mathfrak P\cap\mathfrak R$ has a `name' and can therefore be used to form logic program forms, which means that it is in a sense a ``known'' program. As the framework is designed to compute ``novel'' or ``unknown'' programs in the target domain via analogy-making, \prettyref{equ:PTRT} means that ``known'' target programs can always be computed.
\end{remark}

The following result summarizes some useful consequences of \prettyref{thm:Functional_Proportion_Theorem}.

\begin{corollary} For any source program $P\in\mathfrak P$, target program $R\in\mathfrak R$, and joint programs $Q,S\in\mathfrak P\cap\mathfrak R$, the following proportions hold in $(\mathfrak{P,R})$:
\begin{align*} 
	P\to P^c&\righttherefore R\to R^c\\
	P\to P\cup Q&\righttherefore R\to R\cup Q\\
	P\to (Q\circ P)\circ S&\righttherefore R\to (Q\circ R)\circ S\\
	P\to Q\cdot P\cdot S&\righttherefore R\to Q\cdot R\cdot S\\
	P\to facts(P)&\righttherefore R\to facts(R)\\
	P\to head(P)&\righttherefore R\to head(R)\\
	P\to body(P)&\righttherefore R\to body(R)\\
	P\to LM(P)&\righttherefore R\to LM(R).
\end{align*}
\end{corollary}

The following result is an instance of \citeA[Theorem 28]{Antic22}.

\begin{corollary} For any logic programs $P,Q\in\mathfrak P$ and $R\in\mathfrak R$, we have
\begin{align}\label{equ:reflexivity} 
	&(\mathfrak{P,R})\models P\to P\righttherefore R\to R\quad\text{(inner reflexivity)},\\
	&\mathfrak P\models P\to Q\righttherefore P\to Q\quad\text{(reflexivity)}.
\end{align}
\end{corollary}

\section{Examples}\label{sec:Examples}

In this section, we demonstrate the idea of learning logic programs by analogy via directed logic program proportions by giving some illustrative examples.

\begin{example}\label{exa:a_la_b} Let $A=\{a,b\}$ and $B=\{c,d\}$ be propositional alphabets, and let $\mathfrak P$ and $\mathfrak R$ for the moment be the identical spaces of all propositional programs over $A\cup B$. Consider the following directed equation:
\begin{align}\label{equ:a_la_b} 
	\{a\leftarrow b\}\to \left\{
	\begin{array}{l}
	 	b\\
	  	a\leftarrow b
	\end{array}
	\right\}\righttherefore \{c\leftarrow d\}\to X.
\end{align} Here we have at least two candidates for the solution $S$. 

First, we can say that the second program in \prettyref{equ:a_la_b} is obtained from the first by adding the fact $b$, in which case we expect --- by analogy --- that 
\begin{align}\label{equ:S} 
	S=\left\{
	\begin{array}{l}
	  	b\\
	  	c\leftarrow d
	\end{array}
	\right\}
\end{align} is a solution to \prettyref{equ:a_la_b}. Define the form
\begin{align}\label{equ:G} 
	\mathbf G(Z):=Z\cup\{b\}.
\end{align} Then the computations 
\begin{align}\label{equ:G1} 
	\mathbf G(\{a\leftarrow b\})=\left\{
	\begin{array}{l}
	  	b\\
	  	a\leftarrow b
	\end{array}
\right\} \quad\text{and}\quad\mathbf G(\{c\leftarrow d\})=S
\end{align} show that $S$ is indeed a solution by \prettyref{thm:Functional_Proportion_Theorem}, that is, we have
\begin{align}\label{equ:GG} 
	\{a\leftarrow b\}\to \mathbf G(\{a\leftarrow b\})\righttherefore \{c\leftarrow d\}\to \mathbf G(\{c\leftarrow d\}).
\end{align}

However, what if we separate the two domains by saying that $\mathfrak P$ and $\mathfrak R$ are the spaces of propositional programs over the disjoint alphabets $A$ and $B$, respectively? In this case, $Z\to \mathbf G(Z)$ is no longer a valid justification of \prettyref{equ:GG} as $\{b\}$ in \prettyref{equ:G} is not contained in $\mathfrak P\cap\mathfrak R$. This makes sense since, in this case, the ``solution'' $S$ contains the fact $b$ alien to the target domain $\mathfrak R$. Thus the question is whether we can redefine $\mathbf G$, without using the fact $b$, so that \prettyref{equ:G1} holds. Observe that $b$ is also the body of $a\leftarrow b$, which motivates the following definition:
\begin{align*} 
	\mathbf G'(Z):=Z\cup body(Z).
\end{align*} A simple computation shows that $\mathbf G'$ satisfies
\begin{align*} 
	\mathbf G'(\{a\leftarrow b\})= \left\{
	\begin{array}{l}
		b\\
		a\leftarrow b
	\end{array}
	\right\},
\end{align*} which means that we can compute a solution $S'$ of \prettyref{equ:a_la_b} via \prettyref{thm:Functional_Proportion_Theorem} as
\begin{align*} 
	S':=\mathbf G'(\{c\leftarrow d\})=\left\{
	\begin{array}{l}
	  	d\\
	  	c\leftarrow d
	\end{array}
	\right\}.
\end{align*}
\end{example}

\begin{example}\label{exa:Nat_Plus_List_PlusList} Reconsider the situation in \prettyref{exa:mb{Plus}}, where we have derived the abstract form ${\bf Plus}$ generalizing addition. As a consequence of \prettyref{thm:Functional_Proportion_Theorem}, we have the following directed logic program proportion:\footnote{For simplicity, we omit here the variables $u$ and $x$ from notation, that is, we write $Nat$ and $List$ instead of $Nat(x)$ and $List(u,x)$, respectively.}
\begin{align}\label{equ:Nat_Plus_} 
	Nat\to {\bf Plus}(Nat)\righttherefore List\to {\bf Plus}(List).
\end{align} This proportion formalizes the intuition that ``numbers are to addition what lists are to list concatenation.'' Similarly, we have
\begin{align*} 
	Nat\to {\bf Plus}(Nat)\righttherefore Tree\to {\bf Plus}(Tree).
\end{align*}

Without going into technical details, we want to mention that a similar procedure as in \prettyref{exa:mb{Plus}} applied to a program for multiplication yields a form $\mathbf{Times}(Z\langle q\rangle(\vec x))$ such that $\mathbf{Times}(List(u,x))$ is a program for ``multiplying'' lists, e.g.,
\begin{align*} 
	\mathbf{Times}(List(u,x))\models times([a,a],[b,b],[b,b,b,b]),
\end{align*} where the result $[b,b,b,b]$ is obtained from the input lists by concatenating the second list $[b,b]$ $k$ times with itself, where $k$ is the length of the first list (in this case $k=2$; the actual content of the first list does not matter here). We then have the following directed logic program proportion as an instance of \prettyref{thm:Functional_Proportion_Theorem}:
\begin{align*} 
	Nat\to \mathbf{Times}(Nat)\righttherefore List\to \mathbf{Times}(List).
\end{align*} In other words, addition is to multiplication what list concatenation is to list ``multiplication.''
\end{example}

\begin{example}\label{exa:Nat_Even_Reverse_EvenReverse} In \prettyref{exa:Even}, we have constructed $Even$ from $Nat$ via composition and in \prettyref{exa:mb{Even}}, we have then derived the abstract form $\mathbf{Even}$ generalizing ``evenness.'' As a consequence of \prettyref{thm:Functional_Proportion_Theorem}, we have the following directed logic program proportion:
\begin{align*} 
	Nat\to\mathbf{Even}(Nat)\righttherefore Reverse\to\mathbf{Even}(Reverse).
\end{align*} This shows that the (seemingly unrelated) program for reversing lists of {\em even} length shares the syntactic property of ``evenness'' with the program for constructing the even numbers.
\end{example}

\begin{example}\label{exa:List_Member_Nat_Z} In \prettyref{exa:mb{Member}}, we have derived the abstract form ${\bf Member}$ generalizing ``membership'' and we have asked the following question: What does ``membership'' mean in the numerical domain? We can now state this question formally in the form of the following directed logic program equation:
\begin{align*} 
	List\to Member\righttherefore Nat\to X.
\end{align*} As a consequence of \prettyref{thm:Functional_Proportion_Theorem}, we have the following directed logic program proportion:
\begin{align*} 
	List(u,x)\to {\bf Member}(List(u,x))\righttherefore Nat(u)\to {\bf Member}(Nat(u)),
\end{align*} where ${\bf Member}(Nat(u))$ is the program computing the numerical ``less than'' relation of \prettyref{exa:mb{Member}}.
\todo[inline]{What happens if we replace $Nat(u)$ by $Nat(x)$?}
\end{example}

\section{Related Work}

Arguably, the most prominent (symbolic) model of analogical reasoning to date is \citeS{Gentner83} {\em Structure-Mapping Theory} (or {\em SMT}), first implemented by \citeA{Falklenhainer89}. Our approach shares with Gentner's SMT its symbolic nature. However, while in SMT mappings are constructed with respect to meta-logical considerations --- for instance, Gentner's {\em systematicity principle} prefers connected knowledge over independent facts --- in our framework `mappings' are realized via directed logic program proportions satisfying mathematically well-defined properties.

Formal models of analogical proportions started to appear only recently \cite{Lepage01,Miclet09,Stroppa06}.


The functional-based view in \citeA{Barbot19} is related to our \prettyref{thm:Functional_Proportion_Theorem} on the preservation of functional dependencies across different domains (\prettyref{sec:Functional_proportion_theorem}). Moreover, \citeA[§7.3]{Antic22} contains a brief discussion on the important difference between analogical proportions and categories as studied in abstract algebra.


Heuristic-Driven Theory Projection (HDTP) \cite{Schmidt14} has a similar focus on analogical proportions between logical theories. The critical difference to our approach is that in our framework, we consider the set of {\em all} generalizations of a program, whereas in HDTP only minimally general generalizations (mggs) are considered, that is, there is no notion of ``justification'' in HDTP and proportions are ``'justified'' by mggs only. Another difference is that HDTP is formulated within first-order logic, whereas our framework is formulated within logic programming. The main benefit of restricting the formalism to logic programs (i.e., sets of Horn clauses) is that the rule-like syntactic form of logic programs allows an algebraization via composition and concatenation --- this is not the case for first-order logic. Moreover, the task of finding generalizations is governed by heuristics in HDTP, which has no counterpart in our theory. In a sense, similar to HDTP, our framework can be interpreted as a generalization of classical anti-unification \cite{Plotkin70,Reynolds70}. More precisely, while anti-unification focuses on {\em least general} generalizations of {\em terms}, we are interested here in {\em all} generalizations of {\em programs} (i.e. logical {\em theories}).

Finally, we want to mention the recent work in \citeA{Antic21-4} where a syntactic and algebraic notion of logic program similarity has been introduced via sequential compositions and decompositions as defined in \prettyref{sec:Composition}.

\section{Future Work}

In this paper, we have demonstrated the utility of our framework of directed logic program proportions for learning logic programs by analogy with numerous examples. 

The main task for future research is to develop methods for the algorithmic computation of solutions to directed program equations as defined in this paper (see \prettyref{pseudocode:Sol}). At its core, this requires algebraic methods for logic program decomposition \cite{Antic21-1,Antic21-2} and deconcatenation, which are then used to compute forms generalizing a given program and (characteristic) justifications of a directed proportion. This task turns out to be highly non-trivial even for the propositional case. In fact, the only domains I fully understand at the moment is the 2-valued boolean domain consisting only of two elements 0 and 1 --- and already in that simple case a whole paper is needed to fully describe all solutions \cite{Antic21-3}! For example, in the arithmetical domain of natural numbers with multiplication, computing all solutions even to a single concrete analogical equation is non-trivial: computing all solutions to $20:4::30:x$ requires an 8-page long computation \cite<cf.>[pp.42, Example 66]{Antic22}. Since logic programs are more complicated than booleans or numbers, providing general algorithms for the computation of some or all solutions to analogical logic program equations is highly non-trivial even in the propositional case and far beyond the scope of the current paper. 

This does not mean that the framework is useless for learning --- to the contrary, the paper shows, I hope, quite convincingly that learning of logic programs via solving analogical equations (which appears to be a novel idea) can {\em in principle} be done via solving (directed) logic program equations as proposed in the paper. It is therefore, in a sense, a ``declarative'' paper which shows {\em what} can be done with logic program proportions --- in the future, more ``procedural'' papers will be needed to resolve the issue of {\em how} solutions to equations are to be computed in practice.


Composition and concatenation are interesting operations on programs in their own right and a comparison to other operators for program modularity \cite<cf.>{Bugliesi94,Brogi99} remains as future work. A related question is whether these operations are sufficient for modeling all plausible analogies in logic programming or whether further operations are needed (``completeness''). It is important to emphasize that in the latter case, adding novel operations to the framework does not affect the general formulations of the core definitions.

In this paper, we have restricted ourselves to Horn programs. In the future we plan to adapt our framework to extended classes of programs as, for example, {\em higher-order} \cite<cf.>{Chen93,Miller12} and {\em non-monotonic} logic programming under the stable model or answer set semantics \cite{Gelfond91} and extensions thereof \cite<cf.>{Brewka11}. For this, we will define the composition and concatenation of answer set programs \cite{Antic21-2} which is non-trivial due to negation as failure occurring in rule bodies (and heads).

Finally, a formal comparison of analogical reasoning and learning as defined in this paper with other forms of reasoning and learning, most importantly inductive logic programming \cite{Muggleton91}, is desirable as this line of research may lead to an interesting combination of different learning methods. 

\section{Conclusion}

This paper studied directed analogical proportions between logic programs for logic-based analogical reasoning and learning in the setting of logic programming. This enabled us to compare logic programs possibly across different domains in a uniform way which is crucial for AI-systems. For this, we defined the composition and concatenation of logic programs and showed, by giving some examples, that syntactically similar programs have similar decompositions. This observation led us to the notion of logic program forms which are proper generalizations of logic programs. We then used forms to formalize directed analogical proportions between logic programs --- as an instance of the author's model of analogical proportions --- as a mechanism for deriving novel programs in an ``unknown'' target domain via analogical transfer --- realized by generalization and instantiation --- from a ``known'' source domain.

\section*{Acknowledgments}

We would like to thank the reviewers for their thoughtful and constructive comments, and for their helpful suggestions to improve the presentation of the article.


\section*{Conflict of interest}

The authors declare that they have no conflict of interest.

\section*{Data availability statement}

The manuscript has no data associated.

\if\isdraft1\newpage\fi
\bibliographystyle{theapa}
\bibliography{/Users/christianantic/Bibdesk/Bibliography,/Users/christianantic/Bibdesk/Preprints,/Users/christianantic/Bibdesk/Publications,/Users/christianantic/Bibdesk/Unpublished}
\if\isdraft1
\newpage

\section{Problems}

\begin{problem} The most ambitious goal which emerges from this paper is a full understanding of the logic program proportion relation $P:Q::R:S$ and of proportional logic program equations of the form $P:Q::R:V$...
\todo[inline]{}
\end{problem}

\begin{problem} Given a program $P$ and a form $\mathbf F$, what can be said about solutions to equations of the form
\begin{align*} 
	P=\mathbf F(\vec Z)?
\end{align*} Is it decidable to find solutions?...
\todo[inline]{}
\end{problem}

\begin{problem} Is the set of generalizations of a program recursive?... $Gen(P)$...\todo{def?} 
\todo[inline]{}
\end{problem}

\section{Vision}

In der Praxis waere es zB relevant, folgende p-Gleichung loesen zu koennen:
\begin{align*} 
	Mac\to macOS \righttherefore iPhone\to X.
\end{align*} Wenn wir jetzt $X=iphoneOS$ so berechnen koennten, dass $iphoneOS$ das Gewuenschte macht, dann hat man sich das haendische Uebertragen des Codes erspart!

\fi
\end{document}